\documentclass[submission,copyright,creativecommons]{eptcs}
\providecommand{\event}{EXPRESS/SOS 2014}

\usepackage{amsthm}
\usepackage{amsmath}
\usepackage{amssymb}
\usepackage{times}
\usepackage{txfonts}
\usepackage{stmaryrd}
\usepackage{url}
\usepackage{amsfonts}
\usepackage{code}
\usepackage{mathrsfs}
\DeclareMathAlphabet{\mathpzc}{OT1}{pzc}{m}{it}
\let\mathpzc\mathscr
\let\mathpzc\mathcal
\def\BNF{\ \  | \ \  }
\def\bondi{{\bf bondi}}
\def\ltrans{[\![}
\def\rtrans{]\!]}

\usepackage{prooftree}
\usepackage{macros}

\newtheorem{theorem}{Theorem}[section]
\newtheorem{lemma}[theorem]{Lemma}
\newtheorem{proposition}[theorem]{Proposition}

\newenvironment{definition}[1][Definition]{\begin{trivlist}
\item[\hskip \labelsep {\bfseries #1}]}{\end{trivlist}}

\newenvironment{remark}[1][Remark]{\begin{trivlist}
\item[\hskip \labelsep {\bfseries #1}]}{\end{trivlist}}

\renewcommand{\beq}{\simeq}

\makeatletter
\def \rightarrowfill{\m@th\mathord{\smash-}\mkern-6mu%
  \cleaders\hbox{$\mkern-2mu\mathord{\smash-}\mkern-2mu$}\hfill
  \mkern-6mu\mathord\to}
\makeatother
\makeatletter
\def \Rightarrowfill{\m@th\mathord{\smash=}\mkern-6mu%
  \cleaders\hbox{$\mkern-2mu\mathord{\smash=}\mkern-2mu$}\hfill
  \mkern-6mu\mathord\Rightarrow}
\makeatother
\def \overstackrel#1#2{\mathrel{\mathop{#1}\limits^{#2}}}

\def\lamdo{\Lang_{A,M,D,\mathit{NO}}}

\def\lamco{\Lang_{A,M,C,\mathit{NO}}}

\def\lapdo{\Lang_{A,P,D,\mathit{NO}}}
\def\lapdn{\Lang_{A,P,D,\mathit{NM}}}

\def\lapco{\Lang_{A,P,C,\mathit{NO}}}
\def\lapcn{\Lang_{A,P,C,\mathit{NM}}}

\def\lsmdi{\Lang_{S,M,D,I}}
\def\lsmco{\Lang_{S,M,C,\mathit{NO}}}

\def\lsmci{\Lang_{S,M,C,I}}

\def\lspco{\Lang_{S,P,C,\mathit{NO}}}

\newcommand{\ifte}[4]{{\bf if}\ #1=#2\ {\bf then}\ #3\ {\bf else}\ #4}
\newcommand{\ift}[3]{{\bf if}\ #1=#2\ {\bf then}\ #3}

\newcommand\rn{{\sf rn}}

\title{On the Expressiveness of Intensional Communication}
\author{Thomas Given-Wilson
\institute{INRIA, Paris, France
\footnote{This work has been supported by the project ANR-12-IS02-001 PACE.}
}
\email{thomas.given-wilson@inria.fr}
}

\def\titlerunning{On the Expressiveness of Intensional Communication}
\def\authorrunning{T. Given-Wilson}

\begin{document}
\makeatactive

\maketitle  

\begin{abstract}
The expressiveness of communication primitives has been explored in a common framework
based on the $\pi$-calculus by considering
four features:
{\em synchronism} (asynchronous vs synchronous),
{\em arity} (monadic vs polyadic data),
{\em communication medium} (shared dataspaces vs channel-based),
and {\em pattern-matching} (binding to a name vs testing name equality).
Here pattern-matching is generalised to account for terms with internal structure such as in recent
calculi like Spi calculi, Concurrent Pattern Calculus and Psi calculi.
This paper explores {\em intensionality} upon terms, in particular
communication primitives that can match upon both names and structures.
By means of possibility/impossibility of encodings, this paper shows that
intensionality alone can encode synchronism, arity, communication-medium, and
pattern-matching, yet no combination of these without intensionality can encode any intensional language.
\end{abstract}

\section{Introduction}

The expressiveness of process calculi based upon their choice of communication primitives
has been explored before
\cite{Palamidessi:2003:CEP:966707.966709,journals/iandc/BusiGZ00,DeNicola:2006:EPK:1148743.1148750,G:IC08,GivenWilsonPHD}.
In \cite{G:IC08} this is detailed by examining combinations of four features, namely:
{\em synchronism}, asynchronous versus synchronous;
{\em arity}, monadic versus polyadic;
{\em communication medium}, shared dataspaces versus channels;
and {\em pattern-matching}, purely binding names versus name equality.
These features are able to represent many popular calculi \cite{G:IC08} such as:
asynchronous or synchronous,  monadic or polyadic $\pi$-calculus
  \cite{Milner:1992:CMP:162037.162038,Milner:1992:CMP:162037.162039,milner:polyadic-tutorial};
\Linda~\cite{Gel85};
Mobile Ambients  \cite{DBLP:conf/fossacs/CardelliG98};
$\mu${\sc Klaim} \cite{10.1109/32.685256};
and semantic-$\pi$ \cite{Castagna:2008:SSP:1367144.1367262}.
However, some recent process calculi include communicable primitives that
have structure such as:
Spi calculus \cite{Abadi:1997:CCP:266420.266432} and pattern-matching Spi calculus \cite{Haack:2006:PS:1165126.1165127};
Concurrent Pattern Calculus (CPC) \cite{GivenWilsonGorlaJay10,givenwilson:hal-00987578}, and variations
thereof \cite{GivenWilsonPHD};
and Psi calculi \cite{BJPV11} and sorted Psi calculi \cite{DBLP:conf/tgc/BorgstromGPVP13}.
Indeed (with the exception of Spi calculus) these calculi include communication
primitives that account for the structure of the terms being communicated.

This paper abstracts away from specific calculi %in the style of \cite{G:IC08}
to provide a general account of the expressiveness of {\em intensional}
communication primitives. Here intensionality is an advanced form of pattern-matching
that allows {\em compound} structures of the form $s\bullet t$ to be bound to a single
name, or to have their structure and components be matched in communication.
For example, consider the following processes:
\begin{equation*}
P \ \define\  \oap {} {a\bullet b}\qquad\quad
Q \ \define\  \iap {} {x\bullet y} .Q'\qquad\quad
R \ \define\  \iap {} {z} .R'\qquad\quad
S \ \define\  \iap {} {\pro a\bullet \pro b}.S'
\end{equation*}
where $P$ is an output of the compound $a\bullet b$.
The inputs of $Q$ and $R$ have binding names of the form $x$ in their
patterns $x\bullet y$ and $z$, respectively.
The input of $S$ tests the names $a$ and $b$ for equality and performs no binding.
These process can be combined to form three possible reductions:
\begin{eqnarray*}
P\bnf Q \redar \{a/x,b/y\}Q' &\quad\qquad
P\bnf R \redar \{a\bullet b/z\}R'\quad\qquad&
P\bnf S \redar S'\; .
\end{eqnarray*}
The first matches the structure of the output of $P$ with the input of $Q$ and binds
$a$ and $b$ to $x$ and $y$, respectively, in $Q'$.
The second binds the entire output of $P$ to the single name $z$ in $R'$.
The third matches the structure and names of the output of $P$ with the structure
and names of the input of $\mathit{S,}$ as they match they interact although no binding
or substitution occurs.
This binding of arbitrary structures to a single name, combined with the
name equality testing of the pattern-matching in \cite{G:IC08}
yields a more expressive intensionality in communication.

By generalising the pattern-matching feature to include intensionality the
original sixteen calculi of \cite{G:IC08} are here expanded to twenty-four.
This paper details how all of the eight new calculi are more expressive
than all of the original sixteen.

The key results 
are that intensionality is sufficient to encode:
synchronous communication into asynchronous communication,
polyadic communication into monadic communication,
channel-based communication into dataspace-based communication,
and generalises the original form of pattern-matching.
The more interesting results are the encoding of polyadicity and channel-based communication,
into monadic and dataspace-based communication, as synchronicity into asynchronicity is
straightforward when either channel-based communication, or polyadicity and name-matching are available.
(Encoding lesser pattern-matching into intensionality is trivial.)

In the other direction intensionality is impossible to encode with any other combination of the other
features.
This arises from the complexity of information that can be used to control communication in intensional interactions.
The key to the result can be intuited by considering the encoding of a minimal input process
$S_0=\iap {} x.P$ and a minimal output process $S_1=\oap {} a$.
Their encodings must interact with some arity $i$, i.e. the reduction
$\encode {S_0\bnf S_1}\redar$ is between an input and output of arity $i$.
Now a process $S_2$ can be constructed that outputs greater than $i$ distinct names, along with
a process $S_3$ that exactly matches all these names in a single interaction and evolves to $Q$.
It follows that $\encode{S_0\bnf S_2}$ reduces with some arity~$j$ and also
$\encode{S_2\bnf S_3}$ reduces with some arity $k$.
Now, if $i=j=k$ then at least one name is not being tested for equality in the reduction
of $\encode{S_2\bnf S_3}$ so there exists an $S_4$ that differs from $S_3$ by only that name.
Thus $\encode{S_4}$ reduces with $\encode{S_2}$, but this contradicts a reasonable encoding as $S_4$
does not interact with $S_3$.
If $i\neq j$ or $j\neq k$ then it is possible to show that the encoded process that is
involved in the two different arities must be able to take either reduction.
It follows that the process can take both reductions and reduce with two other encoded process
and this leads to contradiction of the encoding.
Either, it would be that $\encode{S_0\bnf S_2\bnf S_3}\Redar\encode{P\bnf Q}$ which is a contradiction
as $S_0\bnf S_2\bnf S_3\not\Redar P\bnf Q$.
Or, it would be that the encoding introduces divergent computation, which contradicts a reasonable encoding.

The structure of the paper is as follows.
Section~\ref{sec:calculi} introduces the twenty-four calculi considered here.
Section~\ref{sec:encoding} revises the criteria used for encoding and comparing calculi.
Section~\ref{sec:synch} explores synchronism into intensionality.
Section~\ref{sec:arity} details arity into intensionality.
Section~\ref{sec:medium} formalises communication-medium into intensionality and concludes
that intensionality can encode all other features.
Section~\ref{sec:impossible} presents the impossibility of encoding intensionality into
any other non-intensional calculus.
Section~\ref{sec:conclude} concludes, discusses future and related work, and some motivations
for intensional calculi.

\section{Calculi}
\label{sec:calculi}

This section defines the syntax, operational, and behavioural semantics of the calculi
considered here. This relies heavily on the well-known notions developed for the
$\pi$-calculus, the reference framework, and adapts them when necessary to cope with
different features.

\subsection{Syntax}

Assume a countable set of names ${\mathcal N}$ ranged over by $a,b,c,\ldots$. Traditionally in 
$\pi$-calculus-style calculi names are used for channels, input bindings, and output data. However, here these need to be generalised to account for structure. Thus define the {\em terms} denoted with $s,t,\ldots$ to be
\begin{eqnarray*}
s,t &::=& a\BNF s\bullet t\; .
\end{eqnarray*}
Terms can consist of names such as $a$, or of {\em compounds} $s\bullet t$ that combines two terms into one.
The choice of the $\bullet$ as compound operator is similar to Concurrent Pattern Calculus, and also to be clearly distinct from the traditional comma-separated tuples of polyadic calculi.

The input primitives of different languages will exploit different kinds of patterns.
The non-pattern-matching languages will simply use binding names, denoted $x,y,z,\ldots$.
The {\em name-matching} patterns, denoted $m,n,o,\ldots$ and defined by
\begin{eqnarray*}
m,n &::=& x\BNF \pro a
\end{eqnarray*}
consist of either a {\em binding name} $x$, or a {\em name-match} $\pro a$.
Lastly the intensional patterns, denoted $p,q,\ldots$ will also consider structure and are defined by
\begin{eqnarray*}
p,q &::=& m\BNF p\bullet q \; .
\end{eqnarray*}
The binding names $x$ and name-match $\pro a$ are contained in $m$ from the name-matching calculi,
the {\em compound pattern} $p\bullet q$ combines $p$ and $q$ into a single pattern, and are left associative.
The free names and binding names of name-matching and intensional patterns are as expected, taking
the union of sub-patterns for compound patterns. Note that an intensional pattern is well-formed
if and only if all binding names within the pattern are pairwise distinct.
The rest of this paper will only consider well-formed intensional patterns.

The (parametric) syntax for the languages is:
\begin{eqnarray*}
P,Q,R &::=& \zero\BNF OutProc \BNF IN.P \BNF \res n P \BNF P|Q\BNF \ifte s t P Q \BNF *P\BNF\ok\; .
\end{eqnarray*}
The different languages are obtained by replacing the output $OutProc$ and input $IN.P$ with the various definitions.
The rest of the process forms as are usual:
$\zero$ denotes the null process;
restriction $\res n P$ restricts the visibility of $n$ to $P$;
and parallel composition $P|Q$ allows independent evolution of $P$ and $Q$.
The $\ifte s t P Q$ represents conditional equivalence with $\ift s t P$ used when $Q$ is $\zero$.
The~$*P$ represents replication of the process $P$.
Finally, the $\ok$ is used to represent a success process or state, exploited for reasoning about
encodings as in \cite{G:CONCUR08,GivenWilsonPHD}.

This paper considers the possible combinations of four features for communication:
{\em synchronism} (synchronous vs asynchronous),
{\em arity} (monadic vs polyadic data),
{\em communication medium} (message passing vs shared dataspaces),
and {pattern-matching} (simple binding vs name equality vs intensionality).
As a result there exist twenty-four languages denoted as $\Lambda_{s,a,m,p}$ whose generic element is denoted as $\Lang_{\alpha,\beta,\gamma,\delta}$ where:
\begin{itemize}
\item $\alpha = A$ for asynchronous communication, and $\alpha = S$ for synchronous communication.
\item $\beta = M$ for monadic data, and $\beta = P$ for polyadic data.
\item $\gamma = D$ for dataspace-based communication, and $\gamma = C$ for channel-based communications.
\item $\delta = \mathit{NO}$ for no matching capability, $\delta = \mathit{\mathit{NM}}$ for name-matching, and $\delta = I$ for intensionality.
\end{itemize}
For simplicity a dash $-$ will be used when the instantiation of that feature is unimportant.

Thus the syntax of every language is obtained from the following productions:
\begin{equation*}
\begin{array}{rclll}
\Lang_{A,-,-,-}:&&OutProc ::= OUT\\
\Lang_{S,-,-,-}:&&OutProc ::= OUT.P\\
\Lang_{-,M,D,\mathit{NO}}:&&P,Q,R ::=\ldots &IN ::= \iap {} x& OUT::= \oap {} a\\
\Lang_{-,M,D,\mathit{NM}}:&&P,Q,R ::=\ldots &IN ::= \iap {} m& OUT::= \oap {} a\\
\Lang_{-,M,D,I}:&&P,Q,R ::=\ldots &IN ::= \iap {} p& OUT::= \oap {} t\\
\Lang_{-,M,C,\mathit{NO}}:&&P,Q,R ::=\ldots &IN ::= \iap a x& OUT::= \oap a b\\
\Lang_{-,M,C,\mathit{NM}}:&&P,Q,R ::=\ldots &IN ::= \iap a m& OUT::= \oap a b\\
\Lang_{-,M,C,I}:&&P,Q,R ::=\ldots &IN ::= \iap s p& OUT::= \oap s t\\
\Lang_{-,P,D,\mathit{NO}}:&&P,Q,R ::=\ldots &IN ::= \iap {} {\wt x}& OUT::= \oap {} {\wt a}\\
\Lang_{-,P,D,\mathit{NM}}:&&P,Q,R ::=\ldots &IN ::= \iap {} {\wt m}& OUT::= \oap {} {\wt a}\\
\Lang_{-,P,D,I}:&&P,Q,R ::=\ldots &IN ::= \iap {} {\wt p}& OUT::= \oap {} {\wt t}\\
\Lang_{-,P,C,\mathit{NO}}:&&P,Q,R ::=\ldots &IN ::= \iap a {\wt x}& OUT::= \oap a {\wt b}\\
\Lang_{-,P,C,\mathit{NM}}:&&P,Q,R ::=\ldots &IN ::= \iap a {\wt m}& OUT::= \oap a {\wt b}\\
\Lang_{-,P,C,I}:&\qquad&P,Q,R ::=\ldots &IN ::= \iap s {\wt p}\qquad& OUT::= \oap s {\wt t}\; .
\end{array}
\end{equation*}
Here the denotation $\wt \cdot$ represents a sequence of the form $\cdot _1,\cdot _2,\ldots,\cdot _n$ and can be used for names, terms, and both kinds of patterns.
As usual $\iap a {\ldots,x,\ldots}.P$ and $\res x P$ and $\iap {} {x\bullet\ldots} . P$ bind $x$ in $P$.
Observe that in $\iap a {\ldots,\pro b,\ldots} .P$ and $\iap {} {\ldots\bullet\pro b} .P$
neither $a$ nor $b$ bind in $P$, both are free.
The corresponding notions of free and bound names of a process, denoted ${\sf fn}(P)$ and ${\sf bn}(P)$,
are as usual.
Also note that alpha-conversion, denoted $=_\alpha$ is assumed in the usual manner.
Lastly, an input is well-formed if all binding names in that input occur exactly once, this paper shall only consider well-formed inputs.
Finally, the structural equivalence relation $\equiv$ is defined in Figure~\ref{fig:se}.
\begin{figure}
\begin{equation*}
\begin{array}{c}
P\bnf \zero \equiv P
\qquad\qquad
P\bnf Q \equiv Q \bnf P
\qquad\qquad
P\bnf (Q\bnf R) \equiv (P\bnf Q)\bnf R
\vspace{0.2cm} \\
\ifte s t P Q \equiv P\quad s = t
\qquad\qquad
\ifte s t P Q \equiv Q\quad s\neq t
\vspace{0.2cm} \\
P\equiv P'\quad\mbox{if}\ P=_\alpha P'
\qquad\qquad
\res a \zero \equiv \zero
\qquad \res a \res b P\equiv \res b \res a P
\vspace{0.2cm} \\
P\bnf \res a Q\equiv \res a (P\bnf Q)\quad \mbox{if}\ a\notin{\sf fn}(P)
\qquad\qquad
*P\equiv P\bnf *P \; .
\end{array}
\end{equation*}
\caption{Structural equivalence relation.}
\label{fig:se}
\end{figure}

Observe that $\lamco$, $\lapco$, $\lsmco$, and $\lspco$ align with the communication
primitives of the asynchronous/synchronous monadic/polyadic $\pi$-calculus
\cite{Milner:1992:CMP:162037.162038,Milner:1992:CMP:162037.162039,milner:polyadic-tutorial}.
The language $\lapdn$ aligns with \Linda \cite{Gel85};
the languages $\lamdo$ and $\lapdo$ with the monadic/polyadic Mobile Ambients \cite{DBLP:conf/fossacs/CardelliG98};
and $\lapcn$ with that of $\mu${\sc Klaim} \cite{10.1109/32.685256} or semantic-$\pi$ \cite{Castagna:2008:SSP:1367144.1367262}.

The intensional languages do not exactly match any well-known calculi.
Indeed, the combinations of asynchrony and intensionality, or polyadicity and
intensionality have no obvious candidates in the literature.
However,
the language $\lsmdi$ has been mentioned in \cite{GivenWilsonPHD}, %,GivenWilsonGorla13}, 
as a variation of Concurrent Pattern Calculus \cite{GivenWilsonGorlaJay10,GivenWilsonPHD},
and has a behavioural theory as a specialisation of \cite{GivenWilsonGorla13}.
Similarly, the language $\lsmci$ is very similar to pattern-matching Spi calculus
\cite{Haack:2006:PS:1165126.1165127}
and Psi calculi \cite{BJPV11},
albeit with structural channel terms, and without the assertions or the possibility of repeated
binding names in patterns.
There are also similarities between $\lsmci$ and the polyadic synchronous $\pi$-calculus of
\cite{Carbone:2003:EPP:941344.941346}, although the intensionality is limited to the
channel, i.e.~inputs and outputs of the form $\iap s x.P$ and $\oap s a.P$ respectively.

\begin{remark}
\label{rem:leq}
The languages $\Lambda_{s,a,m,p}$ can be easily ordered; in particular
$\Lang_{\alpha_1,\beta_1,\gamma_1,\delta_1}$ can be encoded into
$\Lang_{\alpha_2,\beta_2,\gamma_2,\delta_2}$ if it holds that
$\alpha_1\leq\alpha_2$ and
$\beta_1\leq\beta_2$ and
$\gamma_1\leq\gamma_2$ and
$\delta_1\leq\delta_2$, where $\leq$ is the least reflexive relation satisfying the following axioms:
\begin{equation*}
A \leq S\qquad\qquad M\leq P\qquad\qquad D\leq C\qquad\qquad \mathit{NO}\leq \mathit{NM}\leq I \; .
\end{equation*}
This can be understood as the lesser language variation being a special case of the more general language.
Asynchronous communication is synchronous communication with all output followed by $\zero$.
Monadic communication is polyadic communication with all tuples of arity one.
Dataspace-based communication is channel-based communication with all $k$-ary tuples communicating with channel name $k$.
Lastly, all name-matching communication is intensional communication without any compounds,
and no-matching capability communication is both without any compounds and with only binding names
in patterns.
\end{remark}

\subsection{Operational Semantics}
\label{subsec:op-sem}

The operational semantics of the languages is given here via reductions as in
\cite{milner:polyadic-tutorial,Honda95onreduction-based}.
An alternative style is via a {\em labelled transition system} (LTS) such as \cite{G:IC08}.
Here the reduction based style is to simplify having to define here the (potentially complex)
labels that occur when intensionality is in play. However, the LTS style can be
used for intensional languages \cite{BJPV11,GivenWilsonPHD,GivenWilsonGorla13}, and indeed
captures many\footnote
{Perhaps all of the languages here, although this has not been proven.}
of the languages here \cite{GivenWilsonGorla13}.

Substitutions, denoted $\sigma,\rho,\ldots$, in non-pattern-matching and name-matching
languages are mappings (with finite domain) from names to names. For intensional languages
substitutions are mappings from names to terms.
The application of a substitution $\sigma$ to a pattern $p$ is defined as follows:
\begin{equation*}
\sigma x = \sigma(x)\ \ x\in\mbox{domain}(\sigma)\qquad
\sigma x = x\ \ x\not\in\mbox{domain}(\sigma)\qquad
\sigma \pro x = \pro {(\sigma x)}\qquad
\sigma (p\bullet q) = (\sigma p)\bullet(\sigma q)\; .
\end{equation*}
Where substitution is as usual on names, and on the understanding that the name-match
syntax can be applied to any term by the following definition:
\begin{equation*}
\pro x \define \pro x\qquad\qquad
\pro {(s\bullet t)} \define \pro s\bullet \pro t\; .
\end{equation*}

Given a substitution $\sigma$ and a process $P$, denote with $\sigma P$ the
(capture-avoiding) application of $\sigma$ to $P$ that behaves in the usual manner.
Note that capture can always be avoided by exploiting $\alpha$-equivalence, which can
in turn be assumed due to \cite{UBN07,BP09}.

\renewcommand{\match}[2]{\{#1/\!\!/#2\}}

Interaction between processes is handled by matching some terms $\wt t$ with
some patterns $\wt p$, and possibly also equivalence of channel-names.
This is handled in two parts, the {\em match} rule $\match t p$ of a single term $t$
with a single pattern $p$ to create a substitution $\sigma$.
This is defined as follows:
\begin{eqnarray*}
\begin{array}{rcl}
\match t x &\define& \{t/x\}\\
\match a {\pro a} &\define& \{\}
\end{array}&\qquad\qquad&
\begin{array}{rcl}
\match {s\bullet t} {p\bullet q} &\define& \match s p \cup \match t q\\
\match t p &\mbox{undefined}& \mbox{otherwise.}
\end{array}
\end{eqnarray*}
Any term $t$ can be matched with a binding name $x$ to generate a substitution from the
binding name to the term $\{t/x\}$.
A single name $a$ can be matched with a name-match for that name $\pro a$ to yield the
empty substitution.
A compound term $s\bullet t$ can be matched by a compound pattern $p\bullet q$ when
the components match to yield substitutions $\match s p=\sigma_1$ and $\match t q=\sigma_2$,
the resulting substitution is the unification of $\sigma_1$ and $\sigma_2$.
Observe that since patterns are well-formed, the substitutions of components will always have
disjoint domain.
Otherwise the match is undefined. % (and interaction cannot occur).

The general case is then the {\em poly-match} rule $\pmtch (\wt t ; \wt p)$ that
determines matching of a sequence of terms $\wt t$ with a sequence of patterns $\wt p$,
that is defined below.
\begin{equation*}
\pmtch(;)=\emptyset \qquad
\prooftree \match s p =\sigma_1 \qquad \pmtch (\wt t ; \wt q) =\sigma_2
\justifies \pmtch(s,\wt t;p,\wt q) = \sigma_1\uplus \sigma_2
\endprooftree \; .
\end{equation*}
The empty sequence matches with the empty sequence to produce the empty substitution.
Otherwise when there is a sequence of terms $s,\wt t$ and a sequence of patterns $p,\wt q$,
the first elements are matched $\match s p$ and the remaining sequences use the poly-match rule.
If both are defined and yield substitutions, then the disjoint union $\uplus$ of substitutions is the
result.
(Like the match rule, the disjoint union is ensured by well-formedness of inputs.)
Otherwise the poly-match rule is undefined, for example when a single match fails, or the
sequences are of unequal arity.

Interaction is now defined by the following axiom:
\begin{eqnarray*}
\oap s {\wt t}.P\bnf \iap s {\wt p}.Q &\quad\redar\quad& P\bnf\sigma Q
\qquad\qquad \pmtch(\wt t;\wt p)=\sigma
\end{eqnarray*}
where the $P$s are omitted in the asynchronous languages,
and the $s$'s are omitted for the dataspace-based languages.
The axiom states that when the poly-match of the terms of an output $\wt t$ match with the
patterns of an input $\wt p$ (and in the channel-based setting the output and input are along
the same channel) yields a substitution $\sigma$, then reduce to ($P$ in the synchronous languages
in parallel with) $\sigma$ applied to $Q$.

The general reduction relation $\redar$ is defined as follows:
\begin{equation*}
\begin{array}{c}
\prooftree
\justifies
\oap s {\wt t}.P\bnf \iap s {\wt p}.Q \quad\redar\quad P\bnf\sigma Q
\endprooftree \quad\pmtch(\wt t;\wt p)=\sigma
\vspace{0.2cm}\\
\prooftree P\redar P'
\justifies P\bnf Q \redar P'\bnf Q
\endprooftree \qquad
\prooftree P\redar P'
\justifies \res a P \redar \res a P'
\endprooftree \qquad
\prooftree P\equiv Q\quad Q\redar Q'\quad Q'\equiv P'
\justifies P\redar P'
\endprooftree\quad
\end{array}
\end{equation*}
with $\Redar$ denoting the reflexive, transitive closure of $\redar$.

Lastly, for each language let $\beq$ denote
a reduction-sensitive reference behavioural equivalence for that language, e.g.~a barbed equivalence.
For the non-intensional languages these are already known, either by
their equivalent language in the literature, such as asynchronous/synchronous monadic/polyadic
$\pi$-calculus, or from \cite{G:IC08}.
For the intensional languages the results in \cite{GivenWilsonGorla13} can be used.

\section{Encodings}
\label{sec:encoding}

This section recalls and adapts the definition of valid encodings as well as some
useful theorems (details in \cite{G:CONCUR08}) for formally relating process calculi.
The validity of such criteria in developing expressiveness studies emerges from the
various works \cite{G:IC08,G:DC10,G:CONCUR08}, that have also recently inspired similar works
\cite{LPSS10,Lanese:2010:EPP:2175486.2175506,gla12}. 

An {\em encoding} of a language $\Lang_1$ into another language $\Lang_2$ is a pair
$(\encode\cdot,\renpol)$ where $\encode\cdot$ translates every $\Lang_1$-process into
an $\Lang_2$-process and $\renpol$ maps every name (of the source language) into a tuple
of $k$ names (of the target language), for $k > 0$.
The translation $\encode\cdot$ turns every term of the source language into a term of the
target; in doing this, the translation may fix some names to play a precise r\^ole 
or may translate a single name into a tuple of names. This can be obtained
by exploiting $\renpol$.

Now consider only encodings that satisfy the following properties.
Let a {\em $k$-ary context} $\context C {\_\,_1; \ldots; \_\,_k}$ be a term where $k$
occurrences of $\zero$ are linearly replaced by the holes $\{\_\,_1;
\ldots; \_\,_k\}$ (every one of the $k$ holes must occur once and only once).
Moreover, denote with $\redar_i$ and $\Redar_i$ 
the relations $\redar$ and $\Redar$ in language $\Lang_i$;
denote with $\redar^\omega_i$ an infinite sequence of reductions in $\Lang_i$.
Moreover, we let $\beq_i$ denote the reference behavioural equivalence for language $\Lang_i$.
Also, let $P \suc_i$ mean that there exists $P'$ such that $P \Redar_i P'$ and $P' \equiv P''\bnf \ok$,
for some $P''$.
Finally, to simplify reading, let $S$ range
over processes of the source language (viz., $\Lang_1$) and $T$ range
over processes of the target language (viz., $\Lang_2$).

\begin{definition}[Valid Encoding]
\label{def:ve}
An encoding $(\encode\cdot,\renpol)$ of $\Lang_1$ into $\Lang_2$
is {\em valid} if it satisfies the following five properties:
\begin{enumerate}
\item {\em Compositionality:} for every $k$-ary operator $\op$ of $\Lang_1$
and for every subset of names $N$,
there exists a $k$-ary context $\CopN C \op N {\_\,_1; \ldots; \_\,_k}$ of $\Lang_2$
such that, for all $S_1,\ldots,S_k$ with ${\sf fn}(S_1,\ldots,S_k) = N$, it holds
that $\encode{\op(S_1,\ldots,S_k)} = \CopN C \op N {\encode{S_1};\ldots;\encode{S_k}}$.

\item {\em Name invariance:}
for every $S$ and name substitution $\sigma$, it holds that
$$
\encode{\sigma S}\ \left\{ 
\begin{array}{ll}
\ =\ \sigma'\encode S& \mbox{ if $\sigma $ is injective}\\
\ \beq_2\ \sigma'\encode S  & \mbox{ otherwise}
\end{array}
\right.
$$
where $\sigma'$ is such that 
$\renpol(\sigma(a)) = \sigma'(\renpol(a))$
for every name $a$.

\item {\em Operational correspondence:}
\begin{itemize}
\item for all $S \Redar_1 S'$, it holds that $\encode S \Redar_2 \beq_2 \encode {S'}$;
\item for all $\encode S \Redar_2 T$, there exists $S'$ such that $S \Redar_1\!\! S'$ 
and $T \Redar_2 \beq_2\!\! \encode {S'}$.
\end{itemize}

\item {\em Divergence reflection:}
for every $S$ such that 
$\encode S \redar\!\!_2^\omega$, it holds that 
\linebreak $S$ \mbox{$\redar\!\!_1^\omega$}.

\item {\em Success sensitiveness:}
for every $S$, it holds that $S \suc_1$ if and only if $\encode S \suc_2$.
\end{enumerate}
\end{definition}

Now recall a result concerning valid encodings that is useful for showing
i.e.\ for proving that no valid encoding can exist between
a pair of languages ${\mathcal L}_1$ and ${\mathcal L}_2$.

\begin{proposition}[Proposition 5.5 from \cite{G:CONCUR08}]
\label{prop:deadlock}
Let $\encode\cdot$ be a valid encoding; then, $S \noredar\!\!_1$
implies that $\encode S \noredar\!\!_2$.
\end{proposition}

\section{Synchronism in Intensionality}
\label{sec:synch}

This section proves that intensionality is sufficient to encode synchronicity.
That is, that any language $\Lang_{S,\beta,\gamma,-}$ can be encoded into $\Lang_{A,\beta,\gamma,I}$.

The typical approach is to use channels and a fresh name to signal that the output has
been received and thus encode synchronicity \cite{G:IC08}.
The approach here exploits a fresh name and intensionality, with optional channel-based communication
to encode synchronicity in asynchronicity.
Consider the translation $\xtrans\cdot$
that is the identity on all primitives except for input and output which are as follows:
\begin{eqnarray*}
\xtrans{\iap s {p,\wt p} .P} &\define& \iap s {x\bullet p,\wt p}.(\oap x {x}\bnf \xtrans P )\\
\xtrans{\oap s {t,\wt t} .Q} &\define & \res x(\oap s {x\bullet t,\wt t}\bnf \iap x {\pro x}.\xtrans Q)
\end{eqnarray*}
where $x$ is not in the free names of $s$, $p$, $\wt p$, $t$, $\wt t$, $P$, or $Q$, and where
$\wt p$ and $\wt t$ are omitted in the monadic case, and the channels are omitted in
the dataspace-based communication case.
The input is translated to receive an additional name $x$ and then output this back to the
translated output to signal interaction has occurred.
Similarly the output restricts a fresh name $x$ and then transmits this along with the original term.
The continuation of the output is then placed under an input that only interacts with the fresh name.

\begin{lemma}
\label{lem:minussynch-match}
Given a synchronous input $P$
and a synchronous output $Q$
then $\xtrans P \bnf \xtrans Q \redar$ if and only if $P\bnf Q \redar$.
\end{lemma}
\begin{proof}
The proof is by definition of the poly-match rule.
\end{proof}

\begin{lemma}
\label{lem:minussynch-equiv}
If $P\equiv Q$ then $\xtrans P \equiv \xtrans Q$.
Conversely, if $\xtrans P \equiv \xtrans Q$ then $Q=\xtrans {P'}$ for some $P'\equiv P$.
\end{lemma}
\begin{proof}
Straightforward, from the fact that $\equiv$ acts only on operators that
$\xtrans\cdot$ translates homomorphically.
\end{proof}

\begin{lemma}
\label{minussynch-red}
The translation $\xtrans\cdot$ from $\Lang_{S,\beta,\gamma,-}$ into $\Lang_{A,\beta,\gamma,I}$ preserves
and reflects reductions. That is:
\begin{enumerate}
\item If $P\redar P'$ then there exists $Q$ such that $\xtrans P \redar\redar Q$ and $Q\equiv\xtrans{P'}$;
\item if $\xtrans P\redar Q$ then there exists $Q'$ such that $Q\redar Q'$ and $Q'\equiv\xtrans {P'}$
for some $P'$ such that $P\redar P'$.
\end{enumerate}
\end{lemma}
\begin{proof}
Both parts can be proved by straightforward induction on the judgements $P\redar P'$
and $\xtrans P \redar Q$, respectively. In both cases, the base step is the most interesting
and follows from Lemma~\ref{lem:minussynch-match}, for the second case the step $Q\redar Q'$
is ensured by the definition of the translation and match rule.
The inductive cases where the last rule used is a structural one then rely on
Lemma~\ref{lem:minussynch-equiv}.
\end{proof}

\begin{theorem}
\label{thm:synch}
For every language $\Lang_{S\beta,\gamma,-}$ there is a valid encoding into $\Lang_{A,\beta,\gamma,I}$.
\end{theorem}
\begin{proof}
Compositionality and name invariance hold by construction.
Operational correspondence (with structural equivalence in the place of $\beq$)
and divergence reflection follow from Lemma~\ref{minussynch-red}.
Success sensitiveness can be proved as follows: $P\suc$ means that there exists $P'$ and
$k\geq 0$ such that $P\redar^k P'\equiv P''\bnf \ok$; by exploiting Lemma~\ref{minussynch-red}
$k$ times and Lemma~\ref{lem:minussynch-equiv} obtain that
$\xtrans P \redar^{2 k} \xtrans {P'}\equiv \xtrans {P''}\bnf \ok$, i.e.~that $\xtrans P \suc$.
The converse implication can be proved similarly.
\end{proof}

\section{Arity in Intensionality}
\label{sec:arity}

This section proves that intensionality is sufficient to encode polyadicity.
That is, that any language $\Lang_{\alpha,P,\gamma,-}$ can be encoded into $\Lang_{\alpha,M,\gamma,I}$.

The key to these encodings is the translation of the polyadic input and output forms
into a single pattern or term, respectively.
The translation $\encode\cdot$ is the identity on all forms except the input and output,
which exploit a single reserved name \rn\ (note that such a reserved name can be ensured
by the renaming policy \cite{G:CONCUR08,GivenWilsonGorlaJay10}) and are translated as follows:
\begin{eqnarray*}
\xtrans{\iap s {p_1,\ldots,p_i}.P} &\define&
\iap s {(\ldots(\pro\rn\bullet p_1)\bullet\ldots)\bullet p_i} .\xtrans P\\
\xtrans{\oap s {t_1,\ldots,t_i}.Q} &\define &
\oap s {(\ldots(\rn\bullet t_1)\bullet\ldots)\bullet t_i} .\xtrans Q\; .
\end{eqnarray*}
Where
the $Q$s are omitted in the asynchronous case, and
the $s$'s are omitted in the dataspace-based communication case.

\begin{lemma}
\label{lem:minuspoly-match}
Given a polyadic input $P$ 
and a polyadic output $Q$ 
then $\xtrans P \bnf \xtrans Q \redar$ if and only if $P\bnf Q \redar$.
\end{lemma}
\begin{proof}
The proof is by induction on the arity of the polyadic input.
\end{proof}

\begin{lemma}
\label{lem:minuspoly-equiv}
If $P\equiv Q$ then $\xtrans P \equiv \xtrans Q$.
Conversely, if $\xtrans P \equiv \xtrans Q$ then $Q=\xtrans {P'}$ for some $P'\equiv P$.
\end{lemma}
\begin{proof}
Straightforward, from the fact that $\equiv$ acts only on operators that
$\xtrans\cdot$ translates homomorphically.
\end{proof}

\begin{lemma}
\label{minuspoly-red}
The translation $\xtrans\cdot$ from $\Lang_{\alpha,P,\gamma,-}$ into $\Lang_{\alpha,M,\gamma,I}$ preserves
and reflects reductions. That is:
\begin{enumerate}
\item If $P\redar P'$ then $\xtrans P \redar \xtrans{P'}$;
\item if $\xtrans P\redar Q$ then $Q=\xtrans {P'}$ for some $P'$ such that $P\redar P'$.
\end{enumerate}
\end{lemma}
\begin{proof}
Both parts can be proved by straightforward induction on the judgements $P\redar P'$
and $\xtrans P \redar Q$, respectively. In both cases, the base step is the most interesting
and follows from Lemma~\ref{lem:minuspoly-match}; the inductive cases where the last rule
used is a structural one then rely on Lemma~\ref{lem:minuspoly-equiv}.
\end{proof}

\begin{theorem}
\label{thm:poly}
For every language $\Lang_{\alpha,P,\gamma,-}$ there is a valid encoding into $\Lang_{\alpha,M,\gamma,I}$.
\end{theorem}
\begin{proof}
Compositionality and name invariance hold by construction.
Operational correspondence (with structural equivalence in the place of $\beq$)
and divergence reflection follow from Lemma~\ref{minuspoly-red}.
Success sensitiveness can be proved as follows: $P\suc$ means that there exists $P'$ and
$k\geq 0$ such that $P\redar^k P'\equiv P''\bnf \ok$; by exploiting Lemma~\ref{minuspoly-red}
$k$ times and Lemma~\ref{lem:minuspoly-equiv} obtain that
$\xtrans P \redar^k \xtrans {P'}\equiv \xtrans {P''}\bnf \ok$, i.e.~that $\xtrans P \suc$.
The converse implication can be proved similarly.
\end{proof}

\section{Communication-Medium in Intensionality}
\label{sec:medium}

This section proves that intensionality is sufficient to encode channel-based communication.
That is, that any language $\Lang_{\alpha,\beta,C,-}$ can be encoded into $\Lang_{\alpha,\beta,D,I}$.

Similar to the polyadic into intensional case, the key is in the translation of the
input and output forms.
In general the translation $\xtrans\cdot$ is the identity on all forms except the input
and output which are translated as follows:
\begin{eqnarray*}
\xtrans{\iap s {p,\wt p} . P} &\define& \iap {} {\pro s\bullet p,\wt p} .\xtrans P\\
\xtrans{\oap s {t,\wt t} . Q} &\define & \oap {} {s \bullet t,\wt t} .\xtrans Q
\end{eqnarray*}
where $\wt p$ and $\wt t$ are omitted in the monadic case.
The input is translated into a pattern that compounds a name-match of the channel name with
the pattern.
The output is a simple compounding of the channel name with term
(and the $Q$s are omitted in the asynchronous case).

\begin{lemma}
\label{lem:minuschan-match}
Given a channel-based communication input $P$
and a channel-based communication output $Q$ 
then $\xtrans P \bnf \xtrans Q \redar$ if and only if $P\bnf Q \redar$.
\end{lemma}
\begin{proof}
The proof is trivial by the definition of the poly-match and match rules.
\end{proof}

\begin{lemma}
\label{lem:minuschan-equiv}
If $P\equiv Q$ then $\xtrans P \equiv \xtrans Q$.
Conversely, if $\xtrans P \equiv \xtrans Q$ then $Q=\xtrans {P'}$ for some $P'\equiv P$.
\end{lemma}
\begin{proof}
Straightforward, from the fact that $\equiv$ acts only on operators that
$\xtrans\cdot$ translates homomorphically.
\end{proof}

\begin{lemma}
\label{minuschan-red}
The translation $\xtrans\cdot$ from $\Lang_{\alpha,\beta,C,-}$ into $\Lang_{\alpha,\beta,D,I}$ preserves
and reflects reductions. That is:
\begin{enumerate}
\item If $P\redar P'$ then $\xtrans P \redar \xtrans{P'}$;
\item if $\xtrans P\redar Q$ then $Q=\xtrans {P'}$ for some $P'$ such that $P\redar P'$.
\end{enumerate}
\end{lemma}
\begin{proof}
Both parts can be proved by straightforward induction on the judgements $P\redar P'$
and $\xtrans P \redar Q$, respectively. In both cases, the base step is the most interesting
and follows from Lemma~\ref{lem:minuschan-match}; the inductive cases where the last rule
used is a structural one then rely on Lemma~\ref{lem:minuschan-equiv}.
\end{proof}

\begin{theorem}
\label{thm:chan}
For every language $\Lang_{\alpha,\beta,C,-}$ there is a valid encoding into $\Lang_{\alpha,\beta,D,I}$.
\end{theorem}
\begin{proof}
Compositionality and name invariance hold by construction.
Operational correspondence (with structural equivalence in the place of $\beq$)
and divergence reflection follow from Lemma~\ref{minuschan-red}.
Success sensitiveness can be proved as follows: $P\suc$ means that there exists $P'$ and
$k\geq 0$ such that $P\redar^k P'\equiv P''\bnf \ok$; by exploiting Lemma~\ref{minuschan-red}
$k$ times and Lemma~\ref{lem:minuschan-equiv} obtain that
$\xtrans P \redar^k \xtrans {P'}\equiv \xtrans {P''}\bnf \ok$, i.e.~that $\xtrans P \suc$.
The converse implication can be proved similarly.
\end{proof}

\medskip

This concludes proving that intensionality can encode: synchronicity, polyadicity,
channel-based communication, and name-matching.
Thus any language $\Lang_{\alpha,\beta,\gamma,\delta}$ can be encoded into
$\Lang_{A,M,D,I}$, and so all the intensional languages can encode each other, and
thus are equally expressive.

\begin{theorem}
\label{thm:int-is-king}
Any language $\Lang_{\alpha 1,\beta 1,\gamma 1,\delta 1}$ can be encoded into any language
$\Lang_{\alpha 2,\beta 2,\gamma 2,I}$. That is, intensionality alone is sufficient to encode: synchronicity,
polyadicity, channel-based communication, and pattern-matching.
\end{theorem}
\begin{proof}
When $\alpha 1 \leq \alpha 2$ then trivial by Remark~\ref{rem:leq}, otherwise use Theorem~\ref{thm:synch}.
When $\beta 1 \leq \beta 2$ then trivial by Remark~\ref{rem:leq}, otherwise use Theorem~\ref{thm:poly}.
When $\gamma 1 \leq \gamma 2$ then trivial by Remark~\ref{rem:leq}, otherwise use Theorem~\ref{thm:chan}.
Finally, observe that $\delta 1\leq \delta 2$ always holds and is trivial by Remark~\ref{rem:leq}.
\end{proof}

\section{Impossible Encodings}
\label{sec:impossible}

% For short of long proofs
\newcommand{\shortproof}[2]{#1} % short

This section considers the impossibility of encoding intensionality with any combination of
other properties. That is, that any language $\Lang_{\alpha1,\beta1,\gamma1,I}$ cannot be
encoded into $\Lang_{\alpha2,\beta2,\gamma2,\delta}$ where $\delta\leq \mathit{NM}$.
The key to the proof is to exploit the contractive nature of intensionality;
that an arbitrarily large term can be bound to a single name, and also the
possibility to match an infinite number of names in a single interaction.
These two properties can be exploited to show that any attempt at encoding yields contradiction.

% Put the proof in a command since reviews are conflicted in what they want from it.
% This allows easy instantiation of the proof with different versions (long/short)
% in different places while keeping it consistent.
\newcommand{\theproof}{
\shortproof
{\begin{theorem}
\label{thm:no-int-any-2}
There is no valid encoding of any language $\Lang_{-,-,-,I}$ into $\Lang_{-,-,-,\delta}$
where $\delta \leq \mathit{NM}$.
\end{theorem}}{}
\begin{proof}
For simplicity the proof will show that any encoding from $\Lang_{A,M,D,I}$ into
a language $\Lang_{\alpha,\beta,\gamma,\delta}$ is impossible,
the generalisation follows from the fact that composition of encodings is an encoding.

The proof is by contradiction.
Assume there exists a valid encoding $\xtrans\cdot$ from $\Lang_{A,M,D,I}$ into
$\Lang_{\alpha,\beta,\gamma,\delta}$ where $\delta\leq \mathit{NM}$.
Consider the encoding of the processes $S_0=\iap {} x.\oap {} m$ and $S_1=\oap {} a$.
Clearly $\xtrans{S_0 \bnf S_1}\redar$  since $S_0\bnf S_1\redar$.
There exists a reduction $\xtrans{S_0\bnf S_1}\redar$ that must be between an
input and output that both have (the same) maximal arity $k$.
\shortproof{}{(Observe that when $\beta=M$, i.e.~encoding into a monadic language, then $k$ must be 1.)}

Now define the following processes 
$S_2 \define \oap {} {a_1\bullet\ldots\bullet a_{k+2}}$
and
$S_3 \define \iap {} {\pro{a_1}\bullet\ldots\bullet\pro{a_{k+2}}} . \oap {} m$
where $S_2$ outputs $k+2$ distinct names in a single term, and $S_3$ matches all of these names in
a single intensional pattern.

Since $S_2\bnf S_0\redar$ it must be that $\xtrans{S_2\bnf S_0}\redar$
for the encoding to be valid.
Now consider the maximal arity of the reduction $\xtrans{S_2\bnf S_0}\redar$:
\begin{itemize}
% arity k
\item If the arity is $k$ \shortproof{}{(this must hold for $\beta=M$, i.e.~encoding into a monadic language)}
  consider the reduction $\xtrans{S_2\bnf S_3}\redar$ with the maximal arity $j$ and
  that must exist since $S_2\bnf S_3\redar$.
  Now consider the relationship of $j$ and $k$.
  \begin{enumerate}
  \item If $j=k$ \shortproof{}{(this must hold for $\beta=M$, i.e.~encoding into a monadic language)}
    then the upper bound on the number of names that are matched in the
    reduction is $k+1$\shortproof{.}{ when
    $\gamma=C$ and $\delta=\mathit{NM}$, i.e.~encoding into a channel-based pattern-matching language.
    (The upper bound is $k$ for $\gamma=D$ and $\delta=\mathit{NM}$;
    1 for $\gamma=C$ and $\delta=\mathit{NO}$;
    and 0 for $\gamma=D$ and $\delta=\mathit{NO}$.)}
%    ($\gamma=C$ and $\delta=\mathit{NM}$, i.e.~encoding into a channel-based pattern-matching language)\footnote
%    {The upper bound is $k$ for $\gamma=D$ and $\delta=\mathit{NM}$;
%    1 for $\gamma=C$ and $\delta=\mathit{NO}$;
%    and 0 for $\gamma=D$ and $\delta=\mathit{NO}$.}.}
    Since not all $k+2$ tuples of names from $\renpol (a_i)$ can be matched in the reduction then
    there must be at least one tuple $\renpol (a_i)$ for $i\in\{1,\ldots,k+2\}$ that is not being
    matched in the interaction $\xtrans{S_2 \bnf S_3}\redar $.
    Now construct $S_4$ that differs from $S_3$ only by swapping one such name $a_i$ with
    $m$:\shortproof{ $
    S_4 \define \iap {} {\pro{a_1}\bullet\ldots\pro{a_{i-1}}\bullet\pro m
    \bullet\pro{a_{i+1}}\ldots\pro{a_{k+2}}}. \oap {} {a_i}$.}{
    \begin{equation*}
    S_4 \define \iap {} {\pro{a_1}\bullet\ldots\pro{a_{i-1}}\bullet\pro m
    \bullet\pro{a_{i+1}}\ldots\pro{a_{k+2}}}. \oap {} {a_i}\; .
    \end{equation*}}
    Now consider the context $\context {C^N_|} {\encode{S_2},\encode\cdot}=\encode{S_2\bnf \cdot}$
    where $N=\{\wt a\cup m\}$.
    Clearly neither $\context {C^N_|} {\encode{S_2},\encode{\zero}}\redar$ nor
    $\context {C^N_|} {\encode{S_2},\encode{S_4}}\redar$ as this would contradict
    Proposition~\ref{prop:deadlock}.
    However, since $S_3$ and $S_4$ differ only by the position of one name whose tuple $\renpol(\cdot)$
    does not appear in the
    reduction $\encode{S_2\bnf S_3}\redar$, it follows that the reason 
    $\context {C^N_|} {\encode{S_2},\encode{S_4}}\not\redar$ must be due to a structural
    congruence difference between $\context {C^N_|} {\encode{S_2},\encode{S_3}}$
    and $\context {C^N_|} {\encode{S_2},\encode{S_4}}$.
    Further, by compositionality of the encoding the difference can only be between
    $\encode{S_3}$ and $\encode{S_4}$.
    Since Proposition~\ref{prop:deadlock} ensures that $\encode{S_3}\not\redar$
    and $\encode{S_4}\not\redar$, the only possibility is a structural congruence difference
    between $\encode{S_3}$ and $\encode {S_4}$.
    \shortproof{Now exploiting the substitution $\sigma=\{m/a_i,a_i/m\}$ that when applied to
    $S_4$ makes it $S_3$ yields contradiction.}{

    Now proceed by induction on the structure of the encoded processes to determine their
    structural difference.
    \begin{enumerate}
    \item If the difference is a restriction $\res c$ then observe that restrictions can
      only prevent reductions, not introduce reductions. Since the reduction is prevented
      for $S_4$ it follows that $\encode {S_4}$ must be of the form $\res c R_4$
      where $R_4$ includes an input or output that contains $c$.
      Now consider the substitution $\sigma=\{m/a_i,a_i/m\}$ that when applied to
      $S_4$ makes it $S_3$.
      By name invariance it must hold that there exists $\sigma '$ such that
      $\encode{\sigma S_4}=\sigma ' \encode{S_4} = \encode{S_3}$.
      However, by definition of the application of a substitution $\sigma '$ cannot
      rename $c$ and so $\context {C^N_|} {\encode{S_2},\sigma'\encode{S_4}}\not\redar$
      which means that $\context {C^N_|} {\encode{S_2},\encode{S_3}}\not\redar$
      which yields contradiction.

    \item If the difference is an $\ifte {c_1}{c_2}S T$ construct then consider the relation
      of $c_1$ to $c_2$ that can only depend upon the translation and must relate $a_i$ or $m$.
      As in the previous case, name invariance and the substitution $\sigma=\{m/a_i,a_i/m\}$
      that when applied to $S_4$ makes it $S_3$, can be used to show contradiction.
    \item Otherwise proceed by induction.
    \end{enumerate}}
%
%%%%%%%%%%%%%%%%%%%%%%%%%%%%
%
  \item If $j\neq k$
    \shortproof{}{(then it must be that $\beta=P$, i.e.~encoding into a polyadic language)}
    then we have that \encode{$S_2}$ must be able to interact with both
    arity $k$ and arity $j$.
    That is, $\xtrans{S_2\bnf\cdot}=\context{C^N_|}{\xtrans {S_2},\xtrans\cdot}$
    where $N=\{\wt a\cup m\}$ and that
    $\context{C^N_|}{\xtrans {S_2},\xtrans{S_0}}$ reduces with arity $k$ and
    $\context{C^N_|}{\xtrans {S_2},\xtrans{S_3}}$ reduces with arity $j$.
    Now it is straightforward, if tedious, to show that since $S_0\bnf S_3\not\redar$
    that $\context{C^N_|}{\xtrans {S_2},\xtrans{S_0\bnf S_3}}$ can perform the same initial
    reductions as either
    $\context{C^N_|}{\xtrans {S_2},\xtrans{S_0\bnf \zero}}$ or 
    $\context{C^N_|}{\xtrans {S_2},\xtrans{\zero\bnf S_3}}$
    by exploiting operational correspondence and Proposition~\ref{prop:deadlock}.

    Thus, it can be shown that $\context{C^N_|}{\xtrans {S_2},\xtrans{S_0\bnf S_3}}$ can perform both
    the $k$ arity reduction of $\xtrans{S_2\bnf S_0}\redar$ and
    the $j$ arity reduction of $\xtrans{S_2\bnf S_3}\redar$.
    Now by exploiting the structural congruence rules it follows that neither of these
    initial reductions can prevent the other occurring. \shortproof{}{
    (The only structural congruence rule that could prevent an interaction occurring is
    the $\ifte {c_1}{c_2}P Q$ where $c_1$ or $c_2$ is modified by the first interaction.
    However, this would require that this congruence be under the input involved in the
    interaction, which would then mean the, now prevented, interaction could not occur
    before the other which contradicts the encoding.)}
    Thus,
    $\context{C^N_|}{\xtrans {S_2},\xtrans{S_0\bnf S_3}}$ must be able to do both of these 
    initial reductions in any order. \shortproof{}{

    }Now consider the process $R$ that has performed both of these initial reductions.
    By operational correspondence it must be that $R\not\Redar\beq\encode{\oap {} m\bnf \oap {} m}$ since
    $S_2\bnf S_0\bnf S_3\not\Redar \oap {} m\bnf \oap {} m$.
    Therefore, $R$ must be able to roll-back the initial step with arity $j$;
    i.e~reduce to a state that is equivalent to the reduction not occurring. \shortproof{}{
    (Or the initial step with arity $k$, but either one is sufficient as
    by operational correspondence $R\Redar\beq\encode{\oap {} m\bnf S_3}$.)

    }Now consider how many names are being matched in the initial reduction with arity $j$.
    If $j< k+1$ the technique of differing on one name used in the case of $j=k$ can be used
    to show that this would introduce divergence on the potential roll-back and thus contradict
    a valid encoding. \shortproof{}{
    (The initial reduction could have occurred when one or more names are different, and thus
    would happen anyway and need to be rolled back. Since the roll-back must not change the
    ability of the processes to interact with other processes, this can be shown to lead to
    an infinite reduction sequence, and thus contradict divergence reflection.)

    }Therefore it must be that $j \geq k + 1$.
    Finally, by exploiting name invariance and substitutions like
    $\{(b_1\bullet\ldots\bullet b_{j + 1})/a_1\}$ applied to $S_2$ and $S_3$
    it follows that either $j\geq k + j + 1$ or both $S_2$ and $S_3$ must have
    infinitely many initial reductions which yields divergence.
  \end{enumerate}
% arity not k
\item If the arity is not $k$ then proceed like the second case above.
\end{itemize}
\vspace{-0.4cm}
\end{proof}}

\theproof
%\vspace*{-0.3cm}

The proof above is for the general case,
there are existing proofs in the literature that can be exploited for partial results.
In particular,
the techniques in \cite{Carbone:2003:EPP:941344.941346}, generalised in \cite{G:IC08,G:DC10},
show that languages that allow for an arbitrary number of names to be matched in interaction
cannot be encoded into to languages that can only match a limited number of names in interaction.

\section{Conclusions and Future Work}
\label{sec:conclude}

Intensional communication primitives alone are highly expressive and can encode the
behaviours of synchronous, polyadic, channel-based, and name-matching communication
primitives. Thus, even the least intensional language 
can encode both the greatest non-intensional language 
and all the other intensional languages. 

There are some languages that include intensionality in their operators, both outside of
communication as in the Spi calculus \cite{Abadi:1997:CCP:266420.266432}, or as part of communication as in
Concurrent Pattern Calculus (CPC) \cite{GivenWilsonGorlaJay10,givenwilson:hal-00987578} (and variations thereof
\cite{GivenWilsonPHD})
and Psi calculi \cite{BJPV11,DBLP:conf/tgc/BorgstromGPVP13}.
However, only one variation of CPC matches any of the family of
intensional languages defined here, that being $\lsmdi$ \cite{GivenWilsonPHD}.
The equivalent expressiveness of all intensional languages here makes exploring each
variant less interesting from a theoretical perspective, but also provides assurance
that none is ``better'' than another from an expressiveness perspective.
This also allows the expressiveness results here to be applied to both CPC and Psi calculi.

Future work in this area could include exploring the r\^ole of either symmetry or
logics in communication. Symmetry of primitives would allow for better
understanding of languages in the style of fusion calculus \cite{705654}
or CPC. This may be of particular interest since symmetry has been used to
show separation results, i.e.~impossibility of encodings, from CPC
into many languages represented here, as well as both fusion calculus and Psi calculi
\cite{GivenWilsonGorlaJay10,GivenWilsonGorla13,givenwilson:hal-00987578}.
Alternatively, considering logics that play a r\^ole in communication would allow for
capturing the behaviours of languages like Concurrent Constraint Programming
\cite{Saraswat:1991:SFC:99583.99627} or Psi calculi. Again the r\^ole of logics has been
used to show separation of Psi calculi from CPC~\cite{GivenWilsonGorla13}.

\newcommand{\wrw}[1]{#1}

\wrw{
%\vspace*{-0.3cm}
\subsection*{Related Work}

This section does not attempt to provide a detailed account of all related works
as this would require an entire paper alone.
Instead, some of the more closely related works are referenced here along
with those that provide the best argument for and against the choices made here.
Further, related works involving calculi with intensional communication primitives
are highlighted.

Expressiveness in process calculi and similar languages has been widely explored,
even when focusing mostly upon the choice of communication primitives
\cite{Palamidessi:2003:CEP:966707.966709,journals/iandc/BusiGZ00,Carbone:2003:EPP:941344.941346,Carbone:2003:EPP:941344.941346,DeNicola:2006:EPK:1148743.1148750,Haagensen200885,G:IC08,GivenWilsonGorlaJay10,GivenWilsonPHD,givenwilson:hal-00987578}.
The choice of valid encodings here is that used, sometimes with mild adaptations, in
\cite{G:CONCUR08,G:DC10,GivenWilsonGorlaJay10,DBLP:journals/corr/abs-1011-6436,GivenWilsonPHD,givenwilson:hal-00987578}
and has also inspired similar works \cite{LPSS10,Lanese:2010:EPP:2175486.2175506,gla12}.
However, there are alternative approaches to encoding criteria or comparing expressive power
\cite{Boudol:1989:NAC:101969.101981,Simone85higher-levelsynchronising,Carbone:2003:EPP:941344.941346,Parrow:2008:EPA:1365089.1365211,gla12}.
Further arguments for, and against, the valid encodings here can be found in
\cite{G:CONCUR08,G:DC10,gla12,givenwilson:hal-00987578}.

There are also some results that fit in between the original 16 languages of \cite{G:IC08}
and those presented here with full intensionality.
The polyadic synchronisation $\pi$-calculi \cite{Carbone:2003:EPP:941344.941346}
allows a vector of names in place of the channel name/term considered in \cite{G:IC08} and 
here. This is likely to have similar expressive power to a name-matching polyadic languages,
and can be easily represented in $\lsmci$, with
inputs and outputs of the form $\iap s x.P$ and $\oap s a.P$ respectively.

%Arguably this also includes polyadic synchronisation $\pi$-calculi \cite{Carbone:2003:EPP:941344.941346}.

There are already existing specific results for the intensional process calculi mentioned here.
Concurrent Pattern Calculus (CPC) can homomorphically encode:
$\pi$-calculus, Linda, and Spi Calculus while none of them can encode CPC
\cite{GivenWilsonGorlaJay10,givenwilson:hal-00987578}. Meanwhile fusion calculus and Psi calculi are
unrelated to CPC in that neither can encode CPC, and CPC cannot encode either of
them \cite{GivenWilsonGorlaJay10,GivenWilsonPHD,givenwilson:hal-00987578}.
Similarly Psi calculi can homomorphically encode $\pi$-calculus \cite{BJPV11}
and indirectly many other calculi, or directly using the techniques here.
Impossibility of encoding results for both CPC and Psi calculi into many calculi
can be derived from the results here.

\subsection*{Motivation}

Clearly intensionality provides significant expressiveness when considering process
calculi.
However, there are further motivations for intensional process calculi that this
paper has not attempted to address.

When considering computational expressiveness intensionality proves to increase
expressive power in the sequential setting \cite{Jay2011}.
By allowing for functions that can match on the structure of their arguments
(in the style that patterns can match against terms here), combinatory logics
exist that prove more expressive than $\lambda$-calculus \cite{Jay2011}.
Indeed, when relating sequential computation to process calculi, intensionality
in the latter allows for both expressing intensionality in the former \cite{GW:ICTAC14},
and for more elegantly capturing Turing Machines \cite{givenwilson:hal-00987594}.

Cryptography, protocols, and security have proved motivating for Spi Calculus \cite{Abadi:1997:CCP:266420.266432}
and pattern-matching Spi Calculus \cite{Haack:2006:PS:1165126.1165127}, both of which introduce intensionality, the latter
in communication as considered here. However, the intensionality presented here
is too strong to support encryption (in the style of Spi Calculus) since it allows
cracking of encryption via patterns of the form ${\sf enc}\bullet(\lambda p\bullet \lambda k)$
where $p$ binds to the plaintext and $k$ to the key \cite{Haack:2006:PS:1165126.1165127,GivenWilsonPHD}.

Psi calculi and sorted Psi calculi attempt to present a general framework for
process calculi that can represent many existing process calculi as an instance
of a (sorted) Psi calculi \cite{DBLP:conf/tgc/BorgstromGPVP13}.
This paper has similar goals in that the results here improve understanding of
the relations of many process calculi to one another. Further, the results allow
for considering the most general language $\Lang_{S,P,C,I}$ while also
recognising that any intensional language has equal expressiveness.

Similarly, the motivation for Concurrent Pattern Calculus is to both
generalise the interaction approaches of many process calculi,
and to represent desirable modeling properties such as exchange
\cite{GivenWilsonPHD,givenwilson:hal-00987578}.
Indeed, \cite{GivenWilsonPHD,20110307:getting_the_goods} demonstrate how intensionality can be used to
capture the pattern-matching of functional programming and
data type constraints with more granularity than a type system.

}

\vspace{-0.3cm}

\bibliographystyle{eptcs}
\bibliography{local}

\newpage
\section*{Appendix A}

This appendix contains a more detailed proof for Theorem~\ref{thm:no-int-any-2}.

% For short of long proofs
%\renewcommand{\shortproof}[2]{#1} % short
\renewcommand{\shortproof}[2]{#2} % long

\theproof

\end{document}